\documentclass[letterpaper, 10 pt, conference]{ieeeconf} 
\IEEEoverridecommandlockouts                             
\usepackage{graphics} 
\usepackage{epsfig} 
\usepackage{mathptmx} 
\usepackage{times} 
\usepackage{amsmath} 
\usepackage{amssymb}  
\usepackage{cite}
\usepackage{thmtools}
\usepackage{etoolbox}

\DeclareMathOperator{\diag}{diag} %
\DeclareMathOperator{\tr}{tr} %
\DeclareMathOperator{\im}{Im} %
\DeclareMathOperator{\re}{Re} %
\newtheorem{thm}{Theorem}
\newtheorem{lem}{Lemma}
\newtheorem{rmk}{Remark}
\newtheorem{real}{Realization}

  \title{\LARGE \bf
Preparation of Pure Gaussian States via Cascaded Quantum Systems*
}

\author{Shan Ma$^{1}$,  Matthew J. Woolley$^{1}$, Ian R. Petersen$^{1}$  and Naoki Yamamoto$^{2}$
\thanks{*This work was supported by the Australian Research Council.}
\thanks{$^{1}$School of Engineering and Information Technology, University of New South Wales at the Australian Defence Force Academy, Canberra, Australian Capital Territory 2600, Australia 
       {\tt\small shanma.adfa@gmail.com, M.Woolley@adfa.edu.au,  i.r.petersen@gmail.com‎}}%
 \thanks{$^{2}$Department of Applied Physics and Physico-informatics, Keio University, Yokohama 223-8522, Japan {\tt\small yamamoto@appi.keio.ac.jp‎}}
}

\begin{document}
\maketitle
\thispagestyle{empty}
\pagestyle{empty}

\begin{abstract}

This paper provides an alternative approach to the problem of preparing pure Gaussian states in a linear quantum system. It is shown that any pure Gaussian state can be generated by a cascade of one-dimensional open quantum harmonic oscillators, without any direct interaction Hamiltonians between these oscillators. This is physically advantageous from an experimental point of view. An example on the preparation of two-mode squeezed states is given to illustrate the theory.   
\end{abstract}

\section{INTRODUCTION}
Gaussian states are of great practical importance in quantum information processing and quantum computation~\cite{weedbrook12:rmp,NC10:book,MVGWRN06:prl,LB99:prl}. They possess several distinctive properties that make them stand out from other quantum states. They are commonly encountered in quantum optics laboratories~\cite{YU13:np}. Vacuum states, thermal states, coherent states and squeezed states are all Gaussian states. Moreover, their Gaussian character is preserved under several common experimental operations. Examples include displacement, squeezing and phase rotation. If we trace out a mode from a multipartite Gaussian state, the reduced state is also Gaussian~\cite{BR04:book}. 

Recently, the problem of  preparing pure Gaussian states has been studied in the literature~\cite{KY12:pra}. The approach is based on engineering the dissipation of the system  as a means of quantum state control. For this reason, the approach is often referred to as reservoir engineering~\cite{CPBZ93:prl,PCZ96:prl,WC13:prl,KM11:prl,WC14:arxiv}. Suppose the open quantum system is described by the following Markovian master equation~\cite{BP02:book}: 
\begin{align}\label{MME}
\frac{d}{d t}\hat{\rho} &=-i[\hat{H},\hat{\rho} ]\notag\\
&+\sum\limits_{j=1}^{m}\left(\hat{b}_{j}\hat{\rho} \hat{b}_{j}^{\dag}-\frac{1}{2}\hat{b}_{j}^{\dag}\hat{b}_{j}\hat{\rho} -\frac{1}{2}\hat{\rho} \hat{b}_{j}^{\dag}\hat{b}_{j}\right),
\end{align}
where $\hat{\rho}$ is the density operator, $\hat{H}$ is the (self-adjoint) Hamiltonian of the system  and $\hat{b}_{j}$ is the coupling operator which describes the interaction between the system and the $j$th environment. Then it is shown in~\cite{KY12:pra} that any pure Gaussian state can be uniquely generated by selecting suitable operators $\hat{H}$ and $\hat{b}_{j}$. This is also known in the field of quantum information, where a quantum state can be characterized via a set of operators known as nullifiers~\cite{MFL11:pra}. 

For general reservoir engineering problems including the above-mentioned 
Gaussian case, a common question is how to practically construct the system operators $(\hat{H}, \hat{b}_j)$ satisfying the requirements to realize the desired dissipation for a large number of modes. 
One such construction is the so-called {\it quasi-local} system-reservoir 
interaction, which requires only a few system components 
interacting with the reservoir~\cite{KBDKMZ08:pra,RLMM12:pra,TV12:ptrsa,
IY13:pra}. On the other hand, realization with the {\it cascade 
connection} has been investigated in~\cite{NJD09:siamjco,N10:tac}; 
in this case {\it all} the system components interact with the reservoir fields in a one-way fashion; see also~\cite{P11:auto}. 
A clear advantage of this construction is that the subsystems can be placed at far away sites, which as a result yield a quantum communication channel.

This paper considers the problem of preparing pure Gaussian states in 
a quantum network synthesis setting~\cite{NJD09:siamjco,N10:tac}, using a  cascade realization. 
Our result is that any pure Gaussian state can be generated by engineering 
a cascade of one-dimensional open quantum harmonic oscillators, without  
any direct interaction Hamiltonians 
between these oscillators. 
In addition to the above-mentioned advantage, this pure cascade realization may be easier to implement in practice than other realizations, 
and hence may serve as an alternative for the experimental preparation of 
pure Gaussian states.

\textit{Notation:}  Let $A=[A_{jk}]$ be a matrix (or a vector) whose entries $A_{jk}$ are complex numbers or operators. We define $A^{\dagger}=[A_{kj}^{\ast}]$, $A^{\top}=[A_{kj}]$, $A^{\#}=[A_{jk}^{\ast}]$, where the superscript ${}^{\ast}$ denotes either the complex conjugate of a number or the adjoint of an operator. Clearly, $A^{\dagger}=\left(A^{\#}\right)^{\top}$. $P_{n}$ is a $2n\times 2n$ permutation matrix defined by $P_{n}[
a_{1} \;a_{2} \;\cdots \;a_{2n}]^{\top}=[
a_{1} \;a_{3} \;\cdots \;a_{2n-1} \;a_{2} \;a_{4} \;\cdots \;a_{2n}]^{\top}$ for any column vector $a=[
a_{1} \;a_{2} \;\cdots \;a_{2n}]^{\top}$. 

\section{Preliminaries}
In this section, we first review some relevant properties of pure Gaussian states~\cite{weedbrook12:rmp,WGKWC04:pra}. Then, we provide an introduction to the network synthesis of linear quantum systems~\cite{GJ09:tac, GJ09:cmp, NJD09:siamjco}. 
 
\subsection{Pure Gaussian States}

Consider a bosonic system of $n$ modes. Each mode is characterized by a pair of quadrature field operators $\{\hat{q}_{j}, \hat{p}_{j}\}$, $j=1,2,\cdots,n$. If we collect them into a vector $\hat{\xi}=(\hat{q}_{1},\cdots,\hat{q}_{n}, \hat{p}_{1},\cdots,\hat{p}_{n})^{\top}$, the canonical commutation relations are written as
\begin{align}\label{commutation 1}
\left[\hat{\xi}_{j}, \hat{\xi}_{k}\right]\triangleq \hat{\xi}_{j}\hat{\xi}_{k}-\hat{\xi}_{k}\hat{\xi}_{j}= i\Sigma_{jk},\quad j,k=1,2,\cdots,2n,
\end{align}
where $\Sigma_{jk}$ is the $(j,k)$ element of the $2n\times 2n$ matrix $\Sigma=\begin{bmatrix}
0 & I_{n}\\
-I_{n} &0
\end{bmatrix}$. 

A state is said to be a Gaussian state if its Wigner function (a function defined in the phase space; see e.g.~\cite{weedbrook12:rmp}) is Gaussian, i.e.,
\begin{align*}
W(\xi) = \frac{ \exp[-1/2(\xi-\langle \hat{\xi}\rangle)^{\top}V^{-1}(\xi-\langle \hat{\xi}\rangle)] }{(2\pi)^{n}\sqrt{\det(V)}},
\end{align*}
where $\langle \hat{\xi}\rangle=\tr(\hat{\xi}\hat{\rho})$ is the mean value vector, and $V$ is the  covariance matrix $V=\frac{1}{2}\langle \triangle\hat{\xi}{\triangle\hat{\xi}}^{\top}+(\triangle\hat{\xi}{\triangle\hat{\xi}}^{\top})^{\top} \rangle$, $\triangle\hat{\xi}=\hat{\xi}-\langle \hat{\xi}\rangle$. A Gaussian state is completely characterized by the  mean vector $\langle \hat{\xi}\rangle$ and the covariance matrix $V$. As the mean vector $\langle \hat{\xi}\rangle$ contains no information about entanglement and can be made to vanish via local unitary operations, we will restrict our attention to zero-mean Gaussian states in the sequel~\cite{WGKWC04:pra,GECP03:QIC}. The purity of a Gaussian state is defined by $\mathbb{P}=\tr(\hat{\rho}^{2})=1/\sqrt{2^{2n}\det(V)}$. We see that a Gaussian state is pure if and only if its  covariance matrix $V$ satisfies $2^{2n}\det(V)=1$. A more explicit and useful  parametrization for the covariance matrix of a pure Gaussian state is as follows. 

\begin{lem}[\protect{\cite{SSM88:pra,WGKWC04:pra}}] 
A matrix $V_{p}$  is the covariance matrix of an $n$-mode pure Gaussian state if and only if there exist real symmetric $n\times n$ matrices $X$ and $Y$ with $Y>0$, such that
\begin{align}\label{covariance}
V_{p}=\frac{1}{2}\begin{bmatrix}
Y^{-1} &Y^{-1}X\\
XY^{-1} &XY^{-1}X+Y
\end{bmatrix}.
\end{align}
\end{lem}

This lemma states that a (zero-mean) pure Gaussian state corresponds to a pair of real symmetric matrices $\{X, Y\}$ with $Y>0$. Given a pure Gaussian state, the matrices $X$ and $Y$ can be uniquely determined, and vice versa. This result will be repeatedly used in the following discussions.

\subsection{Network Synthesis of Linear Quantum Systems}
In the SLH framework developed in~\cite{GJ09:tac,GJ09:cmp}, an open quantum system independently coupled to $m$ environmental fields is characterized by a triple
\begin{align*}
G= (\hat{S}, \hat{L}, \hat{H}),
\end{align*}
where $\hat{S}$ is an $m\times m$ unitary scattering matrix, $\hat{L}$ is an $m\times 1$ coupling operator vector and $\hat{H}$ is the Hamiltonian operator. In the following, we  assume that no scattering is involved between the system and the quantum fields, i.e., $\hat{S}=I_{m}$. If we feed the output of an open quantum system  $G_{1}= (\hat{S}_{1}, \hat{L}_{1}, \hat{H}_{1})$ into the input of another open quantum system $G_{2}= (\hat{S}_{2}, \hat{L}_{2}, \hat{H}_{2})$, the SLH model of this cascade system is
\begin{align}\label{SLHformula}
G_{2}\lhd G_{1} =\left(\hat{S}_{2}\hat{S}_{1},\hat{L}_{2}+\hat{S}_{2}\hat{L}_{1}, \hat{H}_{2}+\hat{H}_{1}+\im(\hat{L}_{2}^{\dagger}\hat{S}_{2}\hat{L}_{1})\right).
\end{align}
Here the notation $G_{2}\lhd G_{1}$ denotes the cascade connection of $G_{1}$ and $G_{2}$.

Within this theory, we now consider a cascade of $n$ one-dimensional open quantum harmonic oscillators $G_{j}= (\hat{S}_{j}, \hat{L}_{j}, \hat{H}_{j})$, $j=1,2,\cdots,n$, as shown in Fig.~1. For each oscillator $G_{j}$, the scattering matrix is $\hat{S}_{j}=I_{m}$; the coupling operator has the linear form $\hat{L}_{j}=K_{j}\hat{x}_{j}$, $\hat{x}_{j}=[\hat{q}_{j}, \hat{p}_{j}]^{\top}$, $K_{j}\in \mathbb{C}^{m\times 2}$ and  the Hamiltonian $\hat{H}_{j}$ has the quadratic form $\hat{H}_{j}=\frac{1}{2}\hat{x}_{j}^{\top}R_{j}\hat{x}_{j}$, where $R_{j}$ is a real symmetric $2\times 2$ matrix. 
 
 \begin{figure}[htb] \label{cascade2}
\hspace*{\fill}\includegraphics[height=1.1cm]{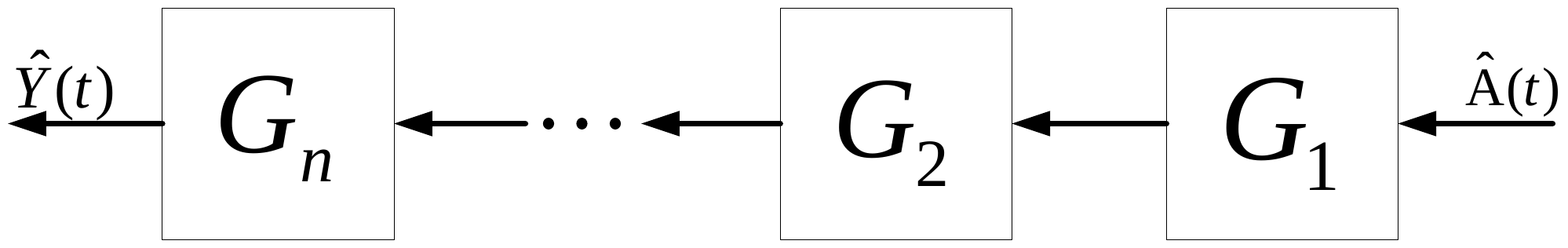} \hspace*{\fill}
\caption{A cascade connection of $n$ one-dimensional open quantum harmonic oscillators $G=G_{n}\lhd\cdots\lhd G_{2}\lhd G_{1}$. The composite system $G$ is a series product of the subsystems $G_{j}$, $j=1,2,\cdots,n$. $\hat{A}(t)$ denotes the input environmental fields; $\hat{Y}(t)$ denotes the output fields.} 
\end{figure} 

In Fig.~1, $\hat{A}(t)=[\hat{A}_{1}(t), \hat{A}_{2}(t),\cdots, \hat{A}_{m}(t)]^{\top}$ is the input of the cascade system $G$. The entries $\hat{A}_{j}(t)$, $j=1,2,\cdots, m$, represent $m$ environmental channels and satisfy the following It\=o rules:
\begin{equation} \label{ito}
\left\{\begin{aligned} 
d\hat{A}_{j}(t)d\hat{A}_{k}^{\ast}(t)&=\delta_{jk}dt,\\
 d\hat{A}_{j}(t)d\hat{A}_{k}(t)&=d\hat{A}_{j}^{\ast}(t)d\hat{A}_{k}^{\ast}(t)=d\hat{A}_{j}^{\ast}(t)d\hat{A}_{k}(t)=0.
 \end{aligned}\right.
\end{equation}
The output field $\hat{Y}(t)=[\hat{Y}_{1}(t), \hat{Y}_{2}(t),\cdots, \hat{Y}_{m}(t)]^{\top}$ results from the interaction between the cascade system $G$ and the input field $\hat{A}(t)$. The entries $\hat{Y}_{j}(t)$, $j=1,2,\cdots, m$, also satisfy It\=o rules similar to~\eqref{ito}~\cite{GJ09:tac,GJ09:cmp}. 

By using the formula~\eqref{SLHformula}, an explicit expression can be derived for the SLH model of the cascade system $G$. 

\begin{lem}[\protect{\cite{N10:tac}}] \label{nurdin}
The system $G=(\hat{S}, \hat{L}, \hat{H})$ obtained by a cascade of $n$ one-dimensional open quantum harmonic oscillators $G_{j}= (\hat{S}_{j}, \hat{L}_{j}, \hat{H}_{j})$, $\hat{S}_{j}=I_{m}$, $j=1,2,\cdots,n$, has the following SLH model:
\begin{equation}
\left\{\begin{aligned}
\hat{S}&=I_{m},\\
\hat{L}&=K\hat{\xi}, \quad K=\begin{bmatrix}
K_{1} &K_{2} &\cdots &K_{n}
\end{bmatrix}P_{n}^{\top},\\
\hat{H}&=\frac{1}{2}\hat{\xi} ^{\top}R \hat{\xi},\quad R=P_{n}MP_{n}^{\top}, 
 \end{aligned}\right. \notag
\end{equation}
where $M=[M_{jk}]_{j,k=1,\cdots,n}$ with $M_{jj}=R_{j}$, $M_{jk}=\im \{K_{j}^{\dagger}K_{k}\}$ whenever $j>k$ and $M_{jk}=M_{kj}^{\top}$ whenever $j<k$. 
\end{lem}

Note that $M_{jk}$ are all real $2\times 2$ matrices. Hence $M$ is a real symmetric $2n\times 2n$ matrix. Alternatively, we can use a quantum stochastic differential equation (QSDE) to  describe the evolution of the entire cascade system $G$. The QSDE model is as follows~\cite{N10:tac,NJD09:siamjco,Y12:ptrsa}: 
\begin{equation} \label{QSDEequation}
\left\{\begin{aligned}
d\hat{\xi}(t)&=\mathcal{A}\hat{\xi}(t)dt+\mathcal{B} \begin{bmatrix}
d\hat{A}(t)\\
d\hat{A}^{\#}(t)
\end{bmatrix},\\
d\hat{Y}(t)&=\mathcal{C}\hat{\xi}(t)dt+\mathcal{D}d\hat{A}(t),
 \end{aligned}\right.
\end{equation}
where $\mathcal{A}=\Sigma(R+\im(K^{\dagger}K))$, $\mathcal{B}=i\Sigma[- K^{\dagger}\; \; K^{\top}]$, $\mathcal{C}=K$, $\mathcal{D}=I_{m}$.

In fact, a bijective correspondence exists between the SLH model and the QSDE model of an open quantum system in the case of $\hat{S}=I$; see~\cite{JNP08:tac} for details. 
%
%

From the QSDE~\eqref{QSDEequation}, it follows that the time evolution of the mean value $\langle \hat{\xi}(t)\rangle$ and the covariance matrix $V(t)$ is as follows:
\begin{align}
\frac{d\langle \hat{\xi}(t)\rangle}{dt}&=\mathcal{A}\langle \hat{\xi}(t)\rangle,\\
\frac{dV(t)}{dt}&=\mathcal{A}V(t)+V(t)\mathcal{A}^{\top}+\frac{1}{2}\mathcal{B}\mathcal{B}^{\dagger}. \label{lyapunov} 
\end{align}

If the initial state of the system is given by a Gaussian state, then 
at any later time $t\geq 0$ the system is in a Gaussian state with mean vector 
$\langle\hat{\xi}(t)\rangle$ and covariance matrix $V(t)$. 
In particular we are interested in a steady Gaussian state with covariance 
matrix $V(\infty)$. In order to generate a pure Gaussian state {\it uniquely}, the matrix $\mathcal{A}$ must be 
a Hurwitz matrix, i.e., every eigenvalue of $\mathcal{A}$ has a negative real part.

\section{Main Result}
In this section, we show that any pure Gaussian state can be generated by engineering a cascade of several one-dimensional open quantum harmonic oscillators. A detailed  construction of such a cascade system is also given. 
Note that there may exist different cascade realizations for a given pure Gaussian state, and some realizations may be easier to implement than others in practice. We provide a feasible construction method here, without any explicit optimization over these constructions. 

\begin{thm} \label{theorem1}
Any $n$-mode pure Gaussian state can be generated by engineering a cascade of $n$ one-dimensional open quantum harmonic oscillators.
\end{thm}
\begin{proof}
We prove this result by construction. Recall that any covariance matrix $V$ of a pure Gaussian state has the representation shown in~\eqref{covariance}. Using the matrices $X$ and $Y$ obtained from~\eqref{covariance}, we construct a cascade system $G=G_{n}\lhd\cdots\lhd G_{2}\lhd G_{1}$ with the SLH model of $G_{j}$, $j=1,2,\cdots,n$, given by
\begin{equation} \label{cascadesubsystem}
\left\{\begin{aligned}
\hat{S}_{j}&= I_{n},\\
\hat{L}_{j}&=K_{j}\hat{x}_{j},\; K_{j}=Y^{-1/2}\left[-(X+iY),\;I_{n}\right]P_{n}\begin{bmatrix}
0_{(2j-2)\times 2}\\
I_{2}\\
0_{(2n-2j)\times 2}
\end{bmatrix},\\
\hat{H}_{j}&=\frac{1}{2}\hat{x}_{j}^{\top}R_{j}\hat{x}_{j},\quad R_{j}=0_{2\times 2}.
 \end{aligned}\right.
\end{equation}

Next, we show that the steady state of this cascade system $G$ is the desired pure Gaussian state with the covariance matrix $V$.
We now calculate the SLH model of the cascade system $G$. 
\begin{equation} 
\left\{\begin{aligned}
\hat{S}&= I_{n},\\
\hat{L}&=K\hat{\xi},\\
\hat{H}&=\frac{1}{2}\hat{\xi}^{\top}R\hat{\xi}.
 \end{aligned}\right.
\end{equation} 

Using Lemma~\ref{nurdin},  
 \begin{align*}
K&=\begin{bmatrix}
K_{1} &K_{2} &\cdots &K_{n}
\end{bmatrix}P_{n}^{\top}\\
&=Y^{-1/2}\left[-(X+iY),\;I_{n}\right]P_{n}\begin{bmatrix}
I_{2} &0_{2\times 2} &\cdots  &0_{2\times 2}\\
0_{2\times 2} &I_{2}  &\cdots  &0_{2\times 2}\\
\vdots &\vdots &\vdots &\vdots   &\vdots\\
0_{2\times 2} &0_{2\times 2} &\cdots  &I_{2}
\end{bmatrix}P_{n}^{\top}\\
&=Y^{-1/2}\left[-(X+iY),\;I_{n}\right].
 \end{align*}
 
On the other hand, when $j> k$, we have
\begin{align*}
&\im(K_{j}^{\dagger}K_{k})\\
&=\im\left([0_{2\times(2j-2)}\; I_{2}\; 0_{2\times (2n-2j)}]P_{n}^{\top}\begin{bmatrix}
-(X-iY)\\
I_{n}
\end{bmatrix}Y^{-1/2}\vphantom{\begin{bmatrix}
0_{(2i-2)\times 2}\\
I_{2}\\
0_{(2n-2i)\times 2}
\end{bmatrix}}\right.\\
&\quad \left. Y^{-1/2}\left[-(X+iY),\;I_{n}\right]P_{n}\begin{bmatrix}
0_{(2k-2)\times 2}\\
I_{2}\\
0_{(2n-2k)\times 2}
\end{bmatrix}\right)\\
&=[0_{2\times(2j-2)}\; I_{2}\; 0_{2\times (2n-2j)}]P_{n}^{\top}\Sigma P_{n}\begin{bmatrix}
0_{(2k-2)\times 2}\\
I_{2}\\
0_{(2n-2k)\times 2}
\end{bmatrix}\\
&=[0_{2\times(2j-2)}\; I_{2}\; 0_{2\times (2n-2j)}]\Theta\begin{bmatrix}
0_{(2k-2)\times 2}\\
I_{2}\\
0_{(2n-2k)\times 2}
\end{bmatrix}\\
&=0, 
\end{align*}
where $\Theta=\diag_{n}(J)$, $J=\begin{bmatrix}
0 &1\\
-1 &0
\end{bmatrix}$. 
Then we have $R=0_{2n\times 2n}$. Therefore, the SLH model of the cascade system $G$ is  
\begin{equation} 
\left\{\begin{aligned}
\hat{S}&= I_{n},\\
\hat{L}&=K\hat{\xi},\quad K=Y^{-1/2}\left[-(X+iY),\;I_{n}\right],\\
\hat{H}&=\frac{1}{2}\hat{\xi}^{\top}R\hat{\xi},\quad R=0_{2n\times 2n}.
 \end{aligned}\right. \notag
\end{equation}

 It follows from the QSDE~\eqref{QSDEequation} that
\begin{align*}
\mathcal{A}&=\Sigma(R+\im(K^{\dagger}K))\\
&=\Sigma \im \left(\begin{bmatrix}
(X-iY)Y^{-1}(X+iY) &-(X-iY)Y^{-1}\\
-Y^{-1}(X+iY) & Y^{-1}
\end{bmatrix}\right)\\
&=\Sigma\Sigma\\
&=-I_{2n},\\
\mathcal{B}&=i\Sigma[- K^{\dagger}\; \; K^{\top}]\\
&=i\Sigma\left[\begin{bmatrix}
-X+iY\\
I
\end{bmatrix}Y^{-1/2}\;\; \begin{bmatrix}
-X-iY\\
I
\end{bmatrix}Y^{-1/2}\right]\\
&=i\begin{bmatrix}
Y^{-1/2} &Y^{-1/2}\\
XY^{-1/2}-iY^{1/2} &XY^{-1/2}+iY^{1/2}
\end{bmatrix}.
\end{align*}

Clearly, $\mathcal{A}$ is a Hurwitz matrix. Furthermore, substituting the matrices $\mathcal{A}$ and $\mathcal{B}$ into Eq.~\eqref{lyapunov} yields 
\begin{align}
&\mathcal{A}V +V \mathcal{A}^{\top}+\frac{1}{2}\mathcal{B}\mathcal{B}^{\dagger}\notag\\
=&-\begin{bmatrix}
Y^{-1} &Y^{-1}X\\
XY^{-1} &XY^{-1}X+Y
\end{bmatrix}+\begin{bmatrix}
Y^{-1} &Y^{-1}X\\
XY^{-1} &XY^{-1}X+Y
\end{bmatrix}\notag\\
=&0. \label{lyapunov2}
\end{align}  

The Hurwitz property of $\mathcal{A}$ and the Lyapunov equation~\eqref{lyapunov2} guarantee that this pure Gaussian state is the unique steady state of the cascade system $G$. 
\end{proof}

\begin{rmk}
In our construction, there is no need for a Hamiltonian contribution; the steady state is only determined by the coupling operators (dissipation). This is analogous to the manner in which a pure Gaussian state  can be specified via its nullifier operators~\cite{MFL11:pra}. Note also that the constructed system $G$ is coupled to $n$ quantum noises (environmental channels). That is, $\hat{L}$ is an $n\times 1$ coupling operator vector. This is sometimes unnecessary. In some cases, an $n$-mode pure Gaussian state may be generated by a cascade system with less than $n$ quantum noises. An example to illustrate this is given in Section~\ref{illustrative example}. 
\end{rmk}

\section{Illustrative Example}\label{illustrative example}
This section studies the preparation problem of the so-called 
{\it two-mode squeezed state}~\cite{weedbrook12:rmp}; 
this is a highly symmetric entangled state which is useful for several 
quantum information protocols such as teleportation. 
It is known that the covariance matrix $V$ of a two-mode squeezed state 
is
\begin{align*}
V=\frac{1}{2}\begin{bmatrix}
\cosh(2\alpha) &\sinh(2\alpha) &0 &0\\
\sinh(2\alpha) &\cosh(2\alpha) &0 &0\\
0 &0 &\cosh(2\alpha) &-\sinh(2\alpha)\\
0 &0 &-\sinh(2\alpha) &\cosh(2\alpha)
\end{bmatrix},
\end{align*}
where $\alpha$ is the squeezing parameter. 
Then it can be calculated from~\eqref{covariance} that 
\begin{align*}
X=0_{2\times 2},\quad Y=\begin{bmatrix}
\cosh(2\alpha) &-\sinh(2\alpha)\\
 -\sinh(2\alpha) &\cosh(2\alpha)
\end{bmatrix}.
\end{align*}
Our objective here is to engineer a cascade of two open quantum  harmonic oscillators $G=G_{2}\lhd G_{1}$ shown in Fig.~2, such that the desired two-mode squeezed state is generated uniquely. \\
 \begin{figure}[htb] \label{two-mode-cascade}
\hspace*{\fill}\includegraphics[height=1.1cm]{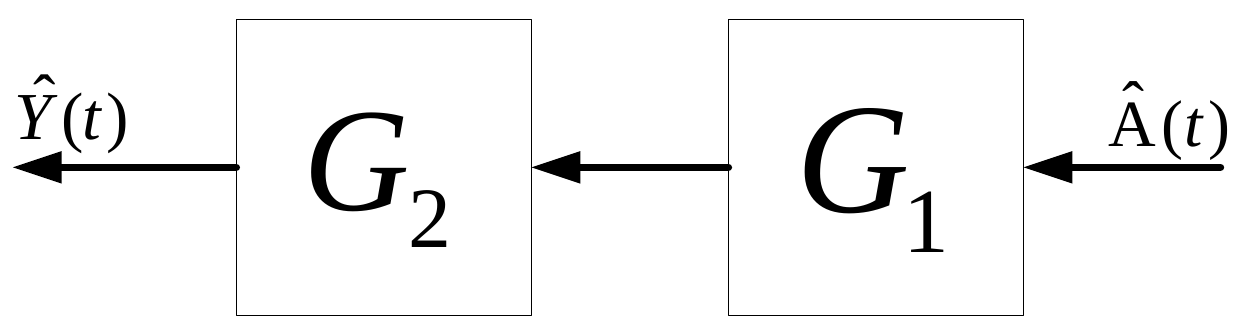} \hspace*{\fill}
\caption{The cascade realization of two-mode squeezed states. The two blocks $G_{1}$ and $G_{2}$ represent two open quantum  harmonic oscillators. The output of $G_{1}$ is fed into the input of $G_{2}$.} 

\end{figure} \\
Here we provide two different realizations.

\begin{real}\label{example1}
In this cascade realization, the SLH models for the subsystems $G_{1}$ and $G_{2}$ are

\begin{equation}
\left\{\begin{aligned}
\hat{S}_{1}&= 1,\\
\hat{L}_{1}&=[iQ_{1}\; 1 ]\hat{x}_{1},,\\
\hat{H}_{1}&=\frac{1}{2}\hat{x}_{1}^{\top}\begin{bmatrix}
2 &Q_{2}\\
Q_{2} &2
\end{bmatrix}\hat{x}_{1},\\
Q_{1}&=\sinh(2\alpha)-\cosh(2\alpha),\\
Q_{2}&=\frac{\sinh^{2}(2\alpha)-\sinh(2\alpha)\cosh(2\alpha)}{\cosh(2\alpha)},
 \end{aligned}\right. \notag
\end{equation}
and
\begin{equation}
\left\{\begin{aligned}
\hat{S}_{2}&= 1,\\
\hat{L}_{2}&=[iQ_{1}\; 1 ]\hat{x}_{2},\\
\hat{H}_{2}&=-\frac{1}{2}\hat{x}_{2}^{\top}\begin{bmatrix}
2 &Q_{2}\\
Q_{2} &2
\end{bmatrix}\hat{x}_{2},\\
 \end{aligned}\right. \notag
\end{equation}
 respectively. 

Using Lemma~\ref{nurdin}, the Hamiltonian matrix $R$ and coupling matrix $K$ of the system $G$ are calculated as follows.
\begin{align*}
R&=\begin{bmatrix}
2 &0 &Q_{2}  &Q_{1}\\
0 &-2 &-Q_{1} &-Q_{2}\\
Q_{2} &- Q_{1} &2  &0\\
Q_{1}  & -Q_{2} &0 &-2
\end{bmatrix},\\
K&=\begin{bmatrix}
iQ_{1} &iQ_{1} &1 &1
\end{bmatrix}.
\end{align*}
Then we have
\begin{align*}
\mathcal{A}&=\Sigma(R+\im(K^{\dagger}K))\\
&=\Sigma\begin{bmatrix}
2 &0 &Q_{2}-Q_{1} &0\\
0 &-2 &-2Q_{1} &-(Q_{2}+Q_{1})\\
Q_{2}+Q_{1} &0 &2  &0\\
2Q_{1} &Q_{1}-Q_{2} &0 &- 2
\end{bmatrix}\\
&=\begin{bmatrix}
Q_{2}+Q_{1} &0 &2 &0\\
2Q_{1} &Q_{1}-Q_{2} &0 &-2\\
-2 &0 &Q_{1}-Q_{2}  &0\\
0 &2 &2Q_{1} &Q_{2}+Q_{1}
\end{bmatrix}\\
\frac{1}{2}\mathcal{B}\mathcal{B}^{\dagger}&=\Sigma\re(K^{\dagger}K)\Sigma^{\top}\\
&=\begin{bmatrix}
1 &1   &0  &0\\
1 &1  &0  &0 \\
0 &0  &Q_{1}^{2} &Q_{1}^{2}\\
0 &0  &Q_{1}^{2} &Q_{1}^{2}\\
\end{bmatrix}
\end{align*}

The characteristic polynomial of $\mathcal{A}$ is
\begin{align*}
\det(\lambda I-\mathcal{A})=\left(\lambda^{2}-2Q_{1}\lambda+Q_{1}^{2}-Q_{2}^{2}+4\right)^{2}
\end{align*}
Because $Q_{1}<0$ and $Q_{1}^{2}-Q_{2}^{2}+4>0$, it follows from Vieta's formulas that $\mathcal{A}$ is Hurwitz. Combining this fact with the Lyapunov equation~\eqref{lyapunov} yields that the two-mode squeezed state is the steady state of the cascade system $G=G_{2}\lhd G_{1}$. Based on the results in~\cite{NJD09:siamjco}, a quantum optical realization of such a cascade system $G$ is given in Fig.~3. As shown in Fig.~3, the Hamiltonian of the system is realized by a pumped nonlinear crystal with specified pump intensity parameter $\varepsilon$ and cavity detuning parameter $\Delta$, and the coupling of the system is realized by implementing an auxiliary cavity. This auxiliary cavity interacts with the system via a cascade of a pumped nonlinear crystal and a beam splitter. The auxiliary cavity has a fast mode that can be adiabatically eliminated.
 \end{real}

 \begin{figure}[htb] \label{two-mode-realization1}
\hspace*{\fill}\includegraphics[width=8.7cm]{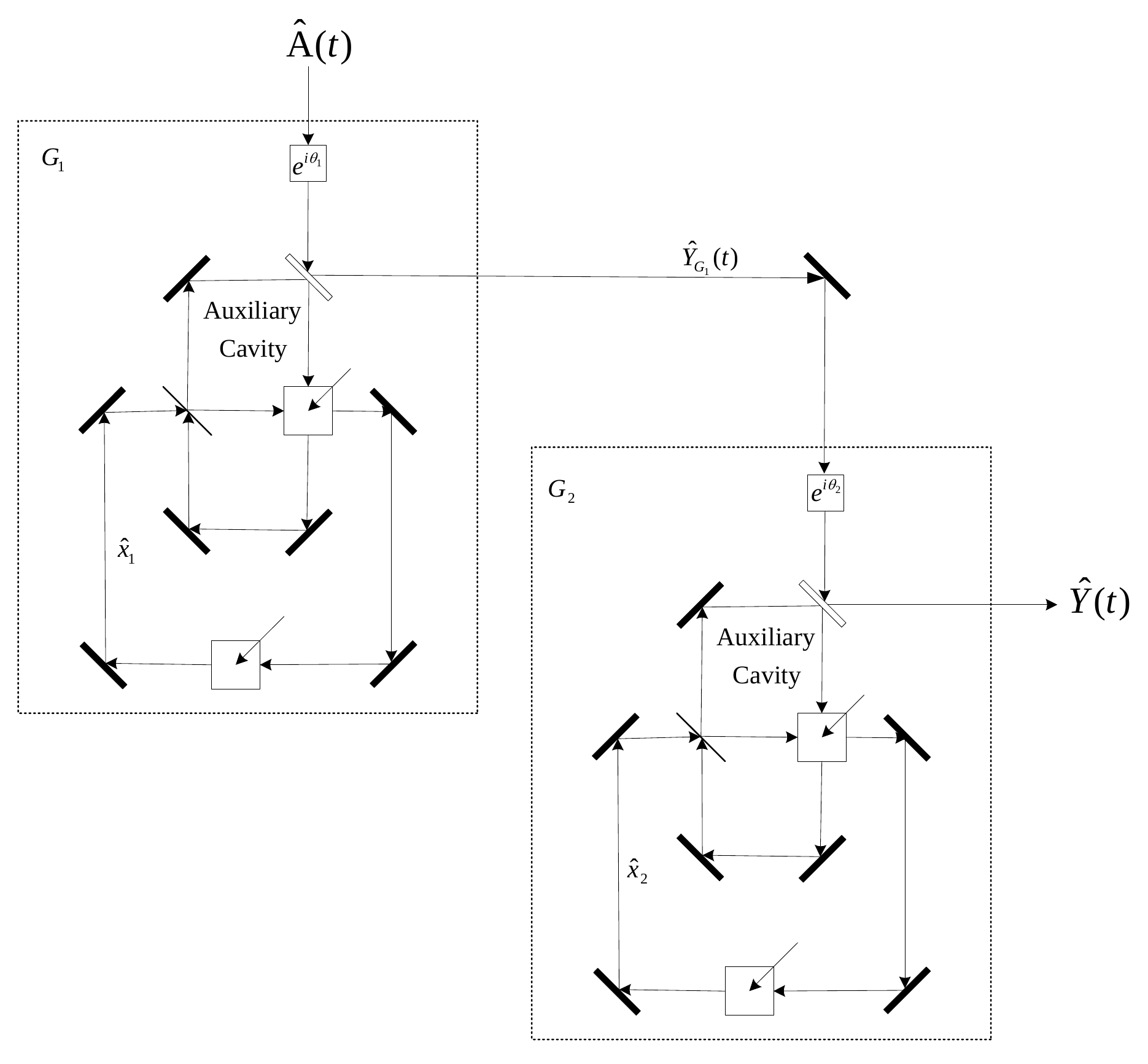} \hspace*{\fill}
\caption{An optical cascade realization of two-mode squeezed states. The output of the subsystem $G_{1}$ is fed into the input of the subsystem $G_{2}$.  
Here the square with the arrow indicates a pumped nonlinear crystal. 
$e^{i\theta}$ indicates a phase shifter. Solid (dark) rectangles denote perfectly reflecting mirrors, while unfilled rectangles denote partially transmitting mirrors. The dark line ``\textbf{$\mathbb{\diagdown}$}'' represents an optical beam splitter.
$\hat{x}_1$ and $\hat{x}_2$ are the optical modes of the bigger cavities in the 
subsystems $G_1$ and $G_2$, respectively; 
 } 
\end{figure} 
\begin{real}\label{example2}

In this cascade realization, we use the result in Theorem~\ref{theorem1}. By direct computations, the SLH models for the subsystems $G_{1}$ and $G_{2}$ are given by

\begin{equation} 
\left\{\begin{aligned}
\hat{S}_{1}&= I_{2},\\
\hat{L}_{1}&=K_{1}\hat{x}_{1},\quad K_{1}=\begin{bmatrix}
-i\cosh(\alpha) &\cosh(\alpha)\\
i\sinh(\alpha) &\sinh(\alpha)
\end{bmatrix}\\
\hat{H}_{1}&=\frac{1}{2}\hat{x}_{1}^{\top}R_{1}\hat{x}_{1},\quad  R_{1}=
0_{2\times 2},
 \end{aligned}\right. \notag
\end{equation}
and
\begin{equation} 
\left\{\begin{aligned}
\hat{S}_{2}&= I_{2},\\
\hat{L}_{2}&=K_{2}\hat{x}_{2},\quad K_{2}=\begin{bmatrix}
i\sinh(\alpha) &\sinh(\alpha)\\
-i\cosh(\alpha) &\cosh(\alpha)
\end{bmatrix},\\
\hat{H}_{2}&=\frac{1}{2}\hat{x}_{2}^{\top}R_{2}\hat{x}_{2},\quad R_{2}=0_{2\times 2},
 \end{aligned}\right. \notag
\end{equation}
 respectively.

A corresponding quantum optical realization of such a cascade system $G=G_{2}\lhd G_{1}$ is given in Fig.~4. As $R_{1}=R_{2}=0_{2\times 2}$, no optical crystals  have to be implemented for the Hamiltonians of the system. On the other hand, as one component of the coupling operator vector $\hat{L}_{1}$ of the subsystem $G_{1}$  is
\begin{align*}
\hat{L}_{11}= \begin{bmatrix}
-i\cosh(\alpha) &\cosh(\alpha)
\end{bmatrix}\begin{bmatrix}
\hat{q}_{1}\\
\hat{p}_{1}
\end{bmatrix}=-i\sqrt{2}\cosh(\alpha)\hat{a}_{1},
\end{align*}
where $\hat{a}_{1}$ denotes the annihilation operator, it can be simply implemented with a partially transmitting mirror without implementing an auxiliary cavity. This situation also occurs in the subsystem $G_{2}$.

  \begin{figure}[htb] \label{two-mode-realization2}
\hspace*{\fill}\includegraphics[width=8.8cm]{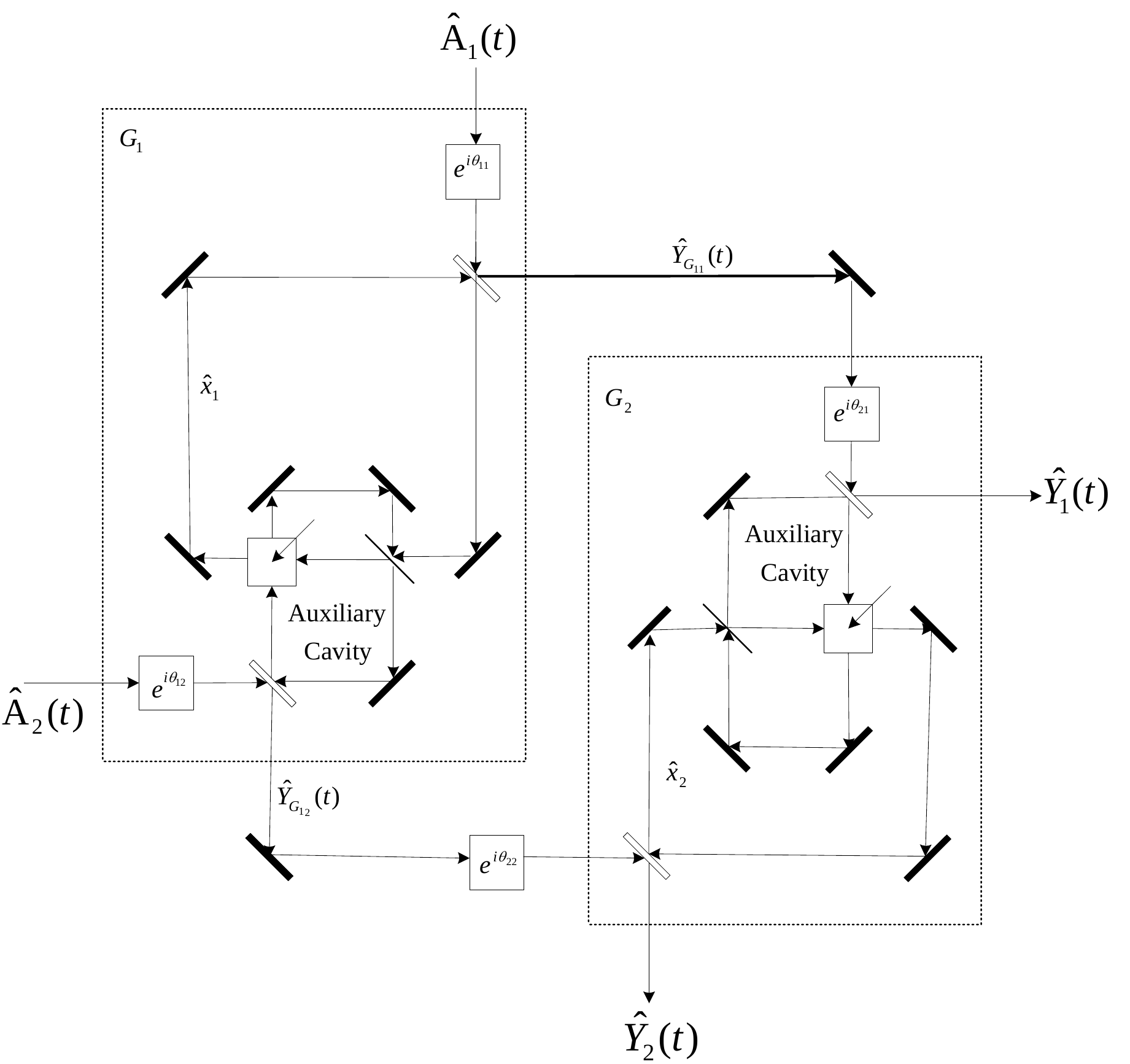} \hspace*{\fill}
\caption{An optical cascade realization of two-mode squeezed states. The output of the subsystem $G_{1}$ is fed into the input of the subsystem $G_{2}$. The meanings of the symbols used here are as given in the caption of Fig.~3.} 
\end{figure}
   \end{real}
   
As shown in Fig.~3 and Fig.~4, the engineered system in Realization~\ref{example1} is coupled to only one quantum noise, while in Realization~\ref{example2}, the engineered system is coupled to two quantum noises. 

\begin{rmk} 
It was shown in~\cite{KY12:pra} that a pure entangled Gaussian 
state is generated in a simpler cascade system composed of standard 
optical parametric oscillators. 
However, the state generated is asymmetric, unlike the two-mode squeezed 
state. 
The point obtained above is that a symmetric (thus highly) 
entangled state is generated in a one-way (thus asymmetric) cascade 
system. 
\end{rmk}

\section{CONCLUSIONS}
This paper has considered the problem of preparing pure Gaussian states in a linear quantum system. We have shown that any pure Gaussian state can be generated by a pure cascade of several one dimensional open quantum  harmonic oscillators. No interaction Hamiltonians have to be implemented between these oscillators. This pure cascade  feature indicates that the proposed approach may be useful in quantum information processing. 
For instance, as mentioned in Section~I, it yields a direct realization of a 
quantum communication channel where each subsystem corresponds to 
a quantum repeater~\cite{HV11:pra}. 
More precisely, we can now dissipatively generate a long-distance entangled 
state by assigning a target pure Gaussian state to 
$\hat{\rho}_{12}\otimes \hat{\rho}_{34} \otimes \cdots \otimes\hat{\rho}_{(n-1)n}$ with $\hat{\rho}_{ij}$ 
a two-mode squeezed state generated among the $i$th and $j$th subsystems 
(repeaters) and then by performing the entanglement swapping via 
Bell-measurement on each site. 


\end{document}